\documentclass{basm}

\newcommand{\case}[1]{\textit{Case $\mathit{#1}$.}\ }

 \author{R.R. Kamalian}

\title{It was not known about simple cycles}

\RunningHead{R.R. Kamalian}%
            {It was not known about simple cycles}



\Abstract{}

\CopyRight { R.R. Kamalian, 2012}

\Address{{\sc R.R. Kamalian}\\
{Institute for Informatics and Automation Problems National Academy
of Sciences of RA, 0014 Yerevan, Republic of Armenia}\\
E-mail: \emph{rrkamalian@yahoo.com}}
\medskip

 \Received { \ August 30, 2011}

\begin{document}

\maketitle

\begin{abstract}
A proper edge $t$-coloring of a graph is a coloring of its edges
with colors $1,2,...,t$ such that all colors are used, and no two
adjacent edges receive the same color. For any integer $n\geq 3$,
all possible values of $t$ are found, for which there exists such a
proper edge $t$-coloring of the simple cycle $C(n)$, which uses for
each pair of adjacent edges either consecutive colors or the first
and the last ones.

\end{abstract}

We consider undirected, simple, finite and connected graphs. For a
graph $G$, we denote by $V(G)$ and $E(G)$ the sets of its vertices
and edges, respectively. The set of edges of $G$ incident with a
vertex $x\in V(G)$ is denoted by $J_G(x)$. For any $x\in V(G)$,
$d_G(x)$ denotes the degree of the vertex $x$ in $G$. For a graph
$G$, $\Delta(G)$ denotes the maximum degree of a vertex of $G$. A
simple cycle with $n$ edges $(n\geq 3)$ is denoted by $C(n)$. A
simple path with $n$ edges $(n\geq 1)$ is denoted by $P(n)$.

For an arbitrary finite set $A$, we denote by $|A|$ the number of
elements of $A$. The set of positive integers is denoted by
$\mathbb{N}$. For any subset $D$ of the set $\mathbb{N}$, we denote
by $D_{(0)}$ and $D_{(1)}$ the subsets of all even and all odd
elements of $D$, respectively.

An arbitrary nonempty subset of consecutive integers is called an
interval. An interval with the minimum element $p$ and the maximum
element $q$ is denoted by $[p,q]$. An interval $D$ is called a
$h$-interval if $|D|=h$.

For any positive integer $k$ define $\varepsilon(k)\equiv
1+\lfloor\frac{k}{2}\rfloor - \lceil\frac{k}{2}\rceil$.

For any nonnegative integer $k$ define
$$
sgn(k)\equiv\left\{
\begin{array}{ll}
0, & \textrm{if $k=0$}\\
1 & \textrm{otherwise.}\\
\end{array}
\right.
$$

A function $\varphi:E(G)\rightarrow [1,t]$ is called a proper edge
$t$-coloring of a graph $G$, if all colors are used, and no two
adjacent edges receive the same color.

The minimum value of $t$ for which there exists a proper edge
$t$-coloring of a graph $G$ is denoted by $\chi'(G)$ \cite{Vizing2}.

If $G$ is a graph, and $\varphi$ is its proper edge $t$-coloring,
where $t\in[\chi'(G),|E(G)|]$, then we define
$U(G,\varphi)\equiv\{e\in E(G)/1<\varphi(e)<t\}$.

If $E_0\subseteq E(G)$, $t\in[\chi'(G),|E(G)|]$, and $\varphi$ is a
proper edge $t$-coloring of a graph $G$, then we set
$\varphi[E_0]\equiv\{\varphi(e)/e\in E_0\}$.

A proper edge $t$-coloring ($t\in[\chi'(G),|E(G)|]$) $\varphi$ of a
graph $G$ is called an interval $t$-coloring of $G$ \cite{Oranj3,
Asratian4, Diss5} if for any $x\in V(G)$, the set $\varphi[J_G(x)]$
is a $d_G(x)$-interval. For any $t\in \mathbb{N}$, we denote by
$\mathfrak{N}_t$ the set of graphs for which there exists an
interval $t$-coloring. Let
$$
\mathfrak{N}=\bigcup_{t\geq 1}\mathfrak{N}_t.
$$

For any $G\in\mathfrak{N}$, we denote by $w_{int}(G)$ and
$W_{int}(G)$ the minimum and the maximum possible value of $t$,
respectively, for which $G\in\mathfrak{N}_t$. For a graph $G$, let
us set $\theta(G)\equiv\{t\in \mathbb{N}/G\in\mathfrak{N}_t\}$.

A proper edge $t$-coloring ($t\in[\chi'(G),|E(G)|]$) $\varphi$ of a
graph $G$ is called a cyclically interval $t$-coloring of $G$, if
for any $x\in V(G)$, at least one of the following two conditions
holds:
\begin{enumerate}
  \item $\varphi[J_G(x)]$ is a $d_G(x)$-interval,
  \item $[1,t]\backslash\varphi[J_G(x)]$ is a $(t-d_G(x))$-interval.
\end{enumerate}

For any $t\in\mathbb{N}$, we denote by $\mathfrak{M}_t$ the set of
graphs for which there exists a cyclically interval $t$-coloring.
Let
$$
\mathfrak{M}\equiv\bigcup_{t\geq 1}\mathfrak{M}_t.
$$

For any $G\in\mathfrak{M}$, we denote by $w_{cyc}(G)$ and
$W_{cyc}(G)$ the minimum and the maximum possible value of $t$,
respectively, for which $G\in\mathfrak{M}_t$. For a graph $G$, let
us set $\Theta(G)\equiv\{t\in \mathbb{N}/G\in\mathfrak{M}_t\}$.

It is clear that for any $G\in\mathfrak{N}$, an arbitrary interval
$t$-coloring ($t\in\theta(G)$) of a graph $G$ is also a cyclically
interval $t$-coloring of $G$. Thus, for any $t\in \mathbb{N}$,
$\mathfrak{N}_t\subseteq\mathfrak{M}_t$ and
$\mathfrak{N}\subseteq\mathfrak{M}$. Let us also note that for an
arbitrary graph $G$, $\theta(G)\subseteq\Theta(G)$. It is also clear
that for any $G\in\mathfrak{N}$, the following inequality is true:
$$
\Delta(G)\leq\chi'(G)\leq w_{cyc}(G)\leq w_{int}(G)\leq
W_{int}(G)\leq W_{cyc}(G)\leq |E(G)|.
$$

In \cite{Diss5, Preprint6}, for any tree $G$, it is proved that
$G\in\mathfrak{N}$, $\theta(G)$ is an interval, and the exact values
of the parameters $w_{int}(G)$, $W_{int}(G)$ are found. In
\cite{Shved1_7}, for any tree $G$, it is proved that
$\Theta(G)=\theta(G)$. Some interesting results on cyclically
interval $t$-colorings and related topics were obtained in
\cite{DeWerra6, DeWerra7, Barth8, Kotzig, Daus9, Csit10}.

In this paper, for any integer $n\geq 3$, it is proved that
$C(n)\in\mathfrak{M}$, and the set $\Theta(C(n))$ is found.

\begin{rem}
Clearly, for any integer $n\geq 3$, $\chi'(C(n))=3-\varepsilon(n)$,
$|E(C(n))|=n$. Therefore, if $t\not\in[3-\varepsilon(n),n]$, then a
proper edge $t$-coloring of $C(n)$ does not exist, and
$C(n)\not\in\mathfrak{N}_t$.
\end{rem}

\begin{rem}\label{rem2}
It is not difficult to see that for any integer $k\geq 2$,
$C(2k)\in\mathfrak{N}$ and $\theta(C(2k))=[2,k+1]$.
\end{rem}

\begin{prop}
For any integer $n\geq 3$, $C(n)\in\mathfrak{M}$,
$n\in\Theta(C(n))$. $\Theta(C(3))=\{3\}$. $\Theta(C(4))=\{2,3,4\}$.
\end{prop}
Proof is trivial.

\begin{thm}
For any integers $n$ and $t$, satisfying the conditions $n\geq 5$
and $t\in[3-\varepsilon(n),n]$, $C(n)\not\in\mathfrak{M}_t$ if and
only if
$t\in[4+\varepsilon(n)\cdot(\frac{n}{2}+\varepsilon(\lfloor\frac{n}{2}\rfloor)-2),n-1]_{(\varepsilon(n))}$.
\end{thm}

\begin{proof}
First let us prove, that if $n\in \mathbb{N}$, $n\geq 5$ and
$t\in[4+\varepsilon(n)\cdot(\frac{n}{2}+\varepsilon(\lfloor\frac{n}{2}\rfloor)-2),n-1]_{(\varepsilon(n))}$,
then $C(n)\not\in\mathfrak{M}_t$.

Assume the contrary: there are $n_0\in \mathbb{N}$, $n_0\geq 5$ and
$t_0\in[4+\varepsilon(n_0)\cdot(\frac{n_0}{2}+\varepsilon(\lfloor\frac{n_0}{2}\rfloor)-2),n_0-1]_{(\varepsilon(n_0))}$,
for which a cyclically interval $t_0$-coloring $\alpha$ of the graph
$C(n_0)$ exists.

Let us construct a graph $H_{00}$ removing from the graph $C(n_0)$
the subset $U(C(n_0),\alpha)$ of its edges. Let us construct a graph
$H_0$ removing from the graph $H_{00}$ all its isolated vertices.

\case{A} $H_0$ is a connected graph.

Let us denote by $F$ the simple path with pendant edges $e'$ and
$e''$ which is isomorphic to the graph $P(n_0-|E(H_0)|+2)$.

\case{A.1} $n_0$ is odd.

Clearly, $t_0\in[4,n_0-1]_{(0)}$. It means that $t_0$ is an even
number, satisfying the inequality $4\leq t_0\leq n_0-1$.

\case{A.1.1} $|E(H_0)|$ is odd.

Clearly, $|E(H_0)|\geq 3$. Since $\alpha$ is a cyclically interval
$t_0$-coloring of $C(n_0)$, we conclude from the definition of
$H_0$, that for a graph $F$, there exists an interval
$(t_0-1)$-coloring $\beta_1$ with $\beta_1(e')=\beta_1(e'')$.
Consequently, the number $n_0-|E(H_0)|+2$ is odd, what contradicts
the same parity of $n_0$ and $|E(H_0)|$.

\case{A.1.2} $|E(H_0)|$ is even.

Clearly, $|E(H_0)|\geq 2$. Since $\alpha$ is a cyclically interval
$t_0$-coloring of $C(n_0)$, we conclude from the definition of
$H_0$, that for a graph $F$, there exists an interval $t_0$-coloring
$\beta_2$ with $\beta_2(e')=1$ and $\beta_2(e'')=t_0$. Consequently,
the number $n_0-|E(H_0)|+2$ is even, what contradicts the different
parity of $n_0$ and $|E(H_0)|$.

\case{A.2} $n_0$ is even.

Clearly,
$t_0\in[\frac{n_0}{2}+2+\varepsilon(\frac{n_0}{2}),n_0-1]_{(1)}$. It
means that $t_0$ is an odd number, satisfying the inequality
$\frac{n_0}{2}+2+\varepsilon(\frac{n_0}{2})\leq t_0\leq n_0-1$.

\case{A.2.1} $|E(H_0)|$ is odd.

Clearly, $|E(H_0)|\geq 3$. Since $\alpha$ is a cyclically interval
$t_0$-coloring of $C(n_0)$, we can conclude from the definition of
$H_0$, that for a graph $F$, there exists an interval
$(t_0-1)$-coloring $\beta_3$ with $\beta_3(e')=\beta_3(e'')$.
Consequently, $n_0>n_0-|E(H_0)|+2=|E(F)|\geq 2t_0-3\geq
n_0+1+2\cdot\varepsilon(\frac{n_0}{2})>n_0$, which is impossible.

\case{A.2.2} $|E(H_0)|$ is even.

Clearly, $|E(H_0)|\geq 2$. Since $\alpha$ is a cyclically interval
$t_0$-coloring of $C(n_0)$, we can conclude from the definition of
$H_0$, that for a graph $F$, there exists an interval $t_0$-coloring
$\beta_4$ with $\beta_4(e')=1$ and $\beta_4(e'')=t_0$. Since $t_0$
is odd, the number $n_0-|E(H_0)|+2$ is also odd, but it is
impossible because of the same parity of $n_0$ and $|E(H_0)|$.

\case{B} $H_0$ is a graph with $m$ connected components, $m\geq 2$.

Assume that:
\begin{enumerate}
  \item $H_1,\ldots,H_m$ are connected components of $H_0$ numbered in
  succession at bypassing of the graph $C(n_0)$ in some fixed direction,
  \item $v_1,\ldots,v_{n_0}$ are vertices of $C(n_0)$ numbered in
  succession at bypassing mentioned in 1),
  \item $e_1,\ldots,e_{n_0}$ are edges of $C(n_0)$ numbered in
  succession at bypassing mentioned in 1),
  \item $v_1\in V(H_1)$, $v_2\in V(H_1)$, $v_{n_0}\not\in V(H_1)$,
  $e_1=(v_1,v_2)$.
\end{enumerate}

Define functions $\zeta:[1,m]\rightarrow[1,n_0-1]$,
$\eta:[1,m]\rightarrow[1,n_0-1]$, $y:[1,2m]\rightarrow\{0,1\}$ as
follows. For any $i\in [1,m]$, set:
$$
\zeta(i)\equiv\min\{k/e_k\in E(H_i)\},\qquad
\eta(i)\equiv\max\{k/e_k\in E(H_i)\}.
$$

For any $j\in[1,2m]$, set
$$
y(j)\equiv\left\{
\begin{array}{ll}
sgn(\alpha(e_{\zeta(\frac{j+1}{2})})-1), & \textrm{if $j$ is odd}\\
sgn(\alpha(e_{\eta(\frac{j}{2})})-1), & \textrm{if $j$ is even.}\\
\end{array}
\right.
$$

Now let us define subgraphs $H'_1,\ldots,H'_m$ of the graph
$C(n_0)$.

For any $i\in[1,m-1]$, let $H'_i$ be the subgraph of $C(n_0)$
induced \cite{West1} by the subset
$\{v_{\eta(i)},v_{\eta(i)+1},\ldots,v_{\zeta(i+1)},v_{\zeta(i+1)+1}\}$
of its vertices. Let $H'_m$ be the subgraph of $C(n_0)$ induced by
the subset $\{v_{\eta(m)},v_{\eta(m)+1},\ldots,v_{n_0},v_1,v_2\}$ of
its vertices.

Let
$$
M_1\equiv\{i\in[1,m]/1\in\alpha[E(H'_i)]\},\qquad
M_2\equiv\{i\in[1,m]/t_0\in\alpha[E(H'_i)]\}.
$$

For any $j\in[1,2m]$, we define a point $\pi_j$ of the
$2$-dimensional rectangle coordinate system by the following way:
$\pi_j\equiv(j,y(j))$.

Let us define a graph $\widetilde{H}$. Set
$V(\widetilde{H})\equiv\{\pi_1,\ldots,\pi_{2m}\}$,
$E(\widetilde{H})\equiv\{(\pi_{2m},\pi_1)\}\cup\{(\pi_j,\pi_{j+1})/j\in[1,2m-1]\}$.
Clearly, $\widetilde{H}\cong C(2m)$.

Let $E_1(\widetilde{H})\equiv\{(\pi_{2q-1},\pi_{2q})/q\in[1,m]\}$,
$E_2(\widetilde{H})\equiv E(\widetilde{H})\backslash
E_1(\widetilde{H})$.

An edge $(\pi',\pi'')$ of the graph $\widetilde{H}$ is called
horizontal if the points $\pi'$ and $\pi''$ have the same ordinate.

Let us denote by $E_{\_}(\widetilde{H})$ the set of all horizontal
edges of the graph $\widetilde{H}$. Set $E_{|}(\widetilde{H})\equiv
E(\widetilde{H})\backslash E_{\_}(\widetilde{H})$. It is easy to
note that the numbers $|E_{\_}(\widetilde{H})|$ and
$|E_{|}(\widetilde{H})|$ are both even.

Now let us define a function $\psi: E(\widetilde{H})\rightarrow
[1,n_0-1]$ by the following way. For an arbitrary $e\in
E(\widetilde{H})$ set:
$$
\psi(e)\equiv\left\{
\begin{array}{ll}
|E(H_q)|, & \textrm{if $e=(\pi_{2q-1},\pi_{2q})$, where $q\in[1,m]$}\\
|E(H'_q)|, & \textrm{if $e=(\pi_{2q},\pi_{2q+1})$, where $q\in[1,m-1]$}\\
|E(H'_m)|, & \textrm{if $e=(\pi_{2m},\pi_1)$.}
\end{array}
\right.
$$

Clearly,
$$
\sum_{e\in E(\widetilde{H})}\psi(e)=n_0+2m.
$$

\case{B.1} $n_0$ is odd.

Clearly, $t_0\in[4,n_0-1]_{(0)}$. It means that $t_0$ is an even
number, satisfying the inequality $4\leq t_0\leq n_0-1$. It is not
difficult to see that in this case, for an arbitrary $e\in
E_{\_}(\widetilde{H})$, $\psi(e)$ is odd, and, moreover, for an
arbitrary $e\in E_{|}(\widetilde{H})$, $\psi(e)$ is even. Since
$|E_{\_}(\widetilde{H})|$ is even, we conclude that the odd number
$\;n_0+2m=\sum_{e\in E_{\_}(\widetilde{H})}\psi(e)+\sum_{e\in
E_{|}(\widetilde{H})}\psi(e)\;$ is represented as a sum of two even
numbers, which is impossible.

\case{B.2} $n_0$ is even.

Clearly,
$t_0\in[\frac{n_0}{2}+2+\varepsilon(\frac{n_0}{2}),n_0-1]_{(1)}$. It
means that $t_0$ is an odd number, satisfying the inequality
$\frac{n_0}{2}+2+\varepsilon(\frac{n_0}{2})\leq t_0\leq n_0-1$. It
is not difficult to see that in this case, for an arbitrary $e\in
E_2(\widetilde{H})\cup E_{\_}(\widetilde{H})$, $\psi(e)$ is odd,
and, moreover, for an arbitrary $e\in E_1(\widetilde{H})\cap
E_{|}(\widetilde{H})$, $\psi(e)$ is even.

\case{B.2.1} $|E_2(\widetilde{H})\cap E_{|}(\widetilde{H})|\geq 2$.

In this case, evidently, there are different integers $i'$ and $i''$
in the set $[1,m]$, for which there exist interval $t_0$-colorings
$\beta'$ and $\beta''$ of the graphs $H'_{i'}$ and $H'_{i''}$,
respectively. Consequently, $n_0=|E(C(n_0))|\geq|E(H'_{i'})\cup
E(H'_{i''})|=|E(H'_{i'})|+|E(H'_{i''})|-|E(H'_{i'})\cap
E(H'_{i''})|\geq|E(H'_{i'})|+|E(H'_{i''})|-2\geq 2t_0-2\geq
n_0+2+2\varepsilon(\frac{n_0}{2})>n_0$, which is impossible.

\case{B.2.2} $|E_2(\widetilde{H})\cap E_{|}(\widetilde{H})|=1$.

Without loss of generality assume that $E_2(\widetilde{H})\cap
E_{|}(\widetilde{H})=\{e^0\}$. Since $|E_{\_}(\widetilde{H})|$ is
even, we conclude that the even number $\;n_0+2m=\sum_{e\in
E_{\_}(\widetilde{H})}\psi(e)+\sum_{e\in
E_{|}(\widetilde{H})}\psi(e)=\sum_{e\in E_2(\widetilde{H})\cap
E_{|}(\widetilde{H})}\psi(e)+\sum_{e\in E_1(\widetilde{H})\cap
E_{|}(\widetilde{H})}\psi(e)+\sum_{e\in
E_{\_}(\widetilde{H})}\psi(e)=\psi(e^0)+\sum_{e\in
E_1(\widetilde{H})\cap E_{|}(\widetilde{H})}\psi(e)+\sum_{e\in
E_{\_}(\widetilde{H})}\psi(e)\;$ is represented as a sum of one odd
and two even numbers, which is impossible.

\case{B.2.3} $|E_2(\widetilde{H})\cap E_{|}(\widetilde{H})|=0$.

Clearly, for any $i\in[1,m]$, the set $\alpha[E(H'_i)]$ contains
exactly one of the colors 1 and $t_0$.

\case{B.2.3.a)} $M_1\neq\emptyset$, $M_2=\emptyset$.

It is not difficult to see that in this case there is $i_1\in M_1$,
for which the set $\alpha[E(H'_{i_1})]$ contains the color $t_0-1$.
It means that there exists an interval $(t_0-1)$-coloring of the
graph $H'_{i_1}$ which colors pendant edges of $H'_{i_1}$ by the
color 1. Consequently, $n_0>|E(H'_{i_1})|\geq 2t_0-3\geq
n_0+1+2\varepsilon(\frac{n_0}{2})>n_0$, which is impossible.

\case{B.2.3.b)} $M_1=\emptyset$, $M_2\neq\emptyset$.

It is not difficult to see that in this case there is $i_2\in M_2$,
for which the set $\alpha[E(H'_{i_2})]$ contains the color 2. It
means that there exists an interval $(t_0-1)$-coloring of the graph
$H'_{i_2}$ which colors pendant edges of $H'_{i_2}$ by the color 1.
Consequently, $n_0>|E(H'_{i_2})|\geq 2t_0-3\geq
n_0+1+2\varepsilon(\frac{n_0}{2})>n_0$, which is impossible.

\case{B.2.3.c)} $M_1\neq\emptyset$, $M_2\neq\emptyset$.

Let us choose $i_3\in M_1$ and $i_4\in M_2$ satisfying the
conditions $|\alpha[E(H'_{i_3})]|=\max_{i\in M_1}|\alpha[E(H'_i)]|$,
$|\alpha[E(H'_{i_4})]|=\max_{i\in M_2}|\alpha[E(H'_i)]|$. Let
$j^{(3)}$ be the maximum color of the set $\alpha[E(H'_{i_3})]$. Let
$j^{(4)}$ be the minimum color of the set $\alpha[E(H'_{i_4})]$.
Clearly, $j^{(3)}\geq j^{(4)}-1$.

It is not difficult to see that there exists an interval
$j^{(3)}$-coloring of the graph $H'_{i_3}$ which colors pendant
edges of $H'_{i_3}$ by the color 1. Hence, $|E(H'_{i_3})|\geq
2j^{(3)}-1$.

It is not difficult to see that there exists an interval
$(t_0-j^{(4)}+1)$-coloring of the graph $H'_{i_4}$ which colors
pendant edges of $H'_{i_4}$ by the color 1. Hence,
$|E(H'_{i_4})|\geq 2\cdot(t_0-j^{(4)}+1)-1=2t_0-2j^{(4)}+1$.

Consequently, we obtain that $n_0>|E(H'_{i_3})\cup
E(H'_{i_4})|=|E(H'_{i_3})|+|E(H'_{i_4})|\geq
2t_0+2(j^{(3)}-j^{(4)})\geq 2t_0-2\geq
n_0+2+2\varepsilon(\frac{n_0}{2})>n_0$, which is impossible.

Thus, we have proved that if $n\in \mathbb{N}$, $n\geq 5$ and
$t\in[4+\varepsilon(n)\cdot(\frac{n}{2}+\varepsilon(\lfloor\frac{n}{2}\rfloor)-2),n-1]_{(\varepsilon(n))}$,
then $C(n)\not\in\mathfrak{M}_t$.

Now let us prove that if $n\in \mathbb{N}$, $n\geq 5$,
$t\in[3-\varepsilon(n),n]$, $C(n)\not\in\mathfrak{M}_t$, then
$t\in[4+\varepsilon(n)\cdot(\frac{n}{2}+\varepsilon(\lfloor\frac{n}{2}\rfloor)-2),n-1]_{(\varepsilon(n))}$.

Assume the contrary. It means that there are $n_0\in \mathbb{N}$,
$n_0\geq 5$, and $t_0\in[3-\varepsilon(n_0),n_0]$, which satisfy the
conditions $C(n_0)\not\in\mathfrak{M}_{t_0}$ and
$t_0\not\in[4+\varepsilon(n_0)\cdot(\frac{n_0}{2}+\varepsilon(\lfloor\frac{n_0}{2}\rfloor)-2),n_0-1]_{(\varepsilon(n_0))}$.

\case{1} $n_0$ is odd.

In this case $t_0\in[3,n_0]$ and $t_0\not\in[4,n_0-1]_{(0)}$, and,
therefore, $t_0\in[3,n_0]_{(1)}$. It means that there exists $m_0\in
\mathbb{N}$, for which $2\leq
m_0=\frac{t_0+1}{2}\leq\frac{n_0+1}{2}$. Let us note that the
equality $m_0=\frac{n_0+1}{2}$ implies $t_0=n_0$, which is
incompatible with the condition $C(n_0)\not\in\mathfrak{M}_{t_0}$.
Hence, $n_0-2m_0\geq1$.

Now, to see a contradiction, it is enough to note that existence of
an interval $t_0$-coloring of a graph $P(2m_0-1)$ with existence of
an interval 2-coloring of a graph $P(n_0-2m_0+1)$ provides existence
of a cyclically interval $t_0$-coloring of the graph $C(n_0)$.

\case{2} $n_0$ is even.

In this case $t_0\in[2,n_0]$ and
$t_0\not\in[\frac{n_0}{2}+2+\varepsilon(\frac{n_0}{2}),n_0-1]_{(1)}$,
and, therefore,
$t_0\in[2,\frac{n_0}{2}+1]\cup([\frac{n_0}{2}+3-\varepsilon(\frac{n_0}{2}),n_0]_{(0)})$.
It follows from Remark \ref{rem2} that
$t_0\in[\frac{n_0}{2}+3-\varepsilon(\frac{n_0}{2}),n_0]_{(0)}$.

Clearly, there exists $m_0\in \mathbb{N}$, $m_0\leq
\frac{1}{2}(\frac{n_0}{2}+\varepsilon(\frac{n_0}{2})-1)$, for which
$t_0=\frac{n_0}{2}+1-\varepsilon(\frac{n_0}{2})+2m_0$. Let us note
that the equality
$m_0=\frac{1}{2}(\frac{n_0}{2}+\varepsilon(\frac{n_0}{2})-1)$
implies $t_0=n_0$, which is incompatible with the condition
$C(n_0)\not\in\mathfrak{M}_{t_0}$. Hence,
$\frac{n_0}{2}+\varepsilon(\frac{n_0}{2})-1-2m_0$ is an even number,
satisfying the inequality
$\frac{n_0}{2}+\varepsilon(\frac{n_0}{2})-1-2m_0\geq 2$.

Now, to see a contradiction, it is enough to note that existence of
an interval $t_0$-coloring of a graph
$P(\frac{n_0}{2}+1-\varepsilon(\frac{n_0}{2})+2m_0)$ with existence
of an interval 2-coloring of a graph
$P(\frac{n_0}{2}+\varepsilon(\frac{n_0}{2})-1-2m_0)$ provides
existence of a cyclically interval $t_0$-coloring of the graph
$C(n_0)$.

Thus, we have proved, that if $n\in \mathbb{N}$, $n\geq5$,
$t\in[3-\varepsilon(n),n]$, $C(n)\not\in\mathfrak{M}_t$, then
$t\in[4+\varepsilon(n)\cdot(\frac{n}{2}+\varepsilon(\lfloor\frac{n}{2}\rfloor)-2),n-1]_{(\varepsilon(n))}$.

It means that we also have

\begin{thm}
For an arbitrary integer $n\geq 5$,
$$
\Theta(C(n))=\left\{
\begin{array}{ll}
[3,n]_{(1)}, & \textrm{if $n$ is odd}\\
{[2,\frac{n}{2}+1]}\cup([\frac{n}{2}+3-\varepsilon(\frac{n}{2}),n]_{(0)}),
& \textrm{if $n$ is even}.
\end{array}
\right.
$$
\end{thm}

\end{proof}

The author thanks P.A. Petrosyan and N.A. Khachatryan for their
attention to this work.

\end{document}